\documentclass[11pt]{amsart}
\usepackage{longtable}
\usepackage{array}
\usepackage{multicol}
\usepackage{graphicx}
\usepackage{amscd}
\usepackage{tikz}
\usepackage{bm}
\usetikzlibrary{arrows.meta,positioning,calc,shapes.geometric}
\usepackage{xcolor}
\usepackage[colorlinks=true,citecolor=black,linkcolor=black,urlcolor=blue]{hyperref}

\makeatletter
 \def\@textbottom{\vskip \z@ \@plus 1pt}
 \let\@texttop\relax
\makeatother
\usepackage{amsmath, amssymb, amsbsy, amsfonts, amsthm, latexsym, amsopn, amstext, amsxtra, euscript, amscd, color, mathrsfs,mathtools}
\usepackage[normalem]{ulem}
\usepackage{soul}

\usepackage{cite}

\makeatletter

\@namedef{subjclassname@2020}{%
  \textup{2020} Mathematics Subject Classification}
\makeatother

\PassOptionsToPackage{hyphens}{url}\usepackage{hyperref}
 
 \usepackage[capbesideposition=outside,capbesidesep=quad]{floatrow}

\restylefloat{table}
\restylefloat{table}
         
\usepackage{multirow,caption}
            
\usepackage{amscd}
\usepackage{color,enumerate}

\newcommand{\RNum}[1]{\lowercase\expandafter{\romannumeral #1\relax}}

\usepackage[colorinlistoftodos,prependcaption,textsize=tiny]{todonotes}

\theoremstyle{plain}
\newtheorem{thm}{Theorem}[section]
\newtheorem{lem}[thm]{Lemma}

\newtheorem{exmp}[thm]{Example}

\newtheorem{rmk}[thm]{Remark}

\newtheorem{thm-con}[thm]{Theorem-Conjecture}
\numberwithin{equation}{section}

\theoremstyle{definition}
\newtheorem{defn}[thm]{Definition}

\def\GL{\operatorname{GL}}

\newcommand{\F}{\mathbb F}

\begin{document}

\title[Constructions of complete permutations over $\F_q^n$]{Constructions of complete permutations over $\F_q^n$}
\author[S. U. Hasan]{Sartaj Ul Hasan}
 \address{Department of Mathematics, Indian Institute of Technology Jammu, Jammu 181221, India}
  \email{sartaj.hasan@iitjammu.ac.in}
  
  \author[R. Kaur]{Ramandeep Kaur}
  \address{Department of Mathematics, Indian Institute of Technology Jammu, Jammu 181221, India}
  \email{2022rma0027@iitjammu.ac.in}
  \author[H. Kumar]{Hridesh Kumar}
\address{Department of Mathematics, Indian Institute of Technology Jammu, Jammu 181221, India}
\email{2021rma2022@iitjammu.ac.in}
 \author[D. Panario]{Daniel Panario}
\address{School of Mathematics and Statistics, Carleton University, Ottawa ON K1S 5B6, Canada}
\email{daniel@math.carleton.ca}
\author[Q. Wang]{Qiang Wang}
\address{School of Mathematics and Statistics, Carleton University, Ottawa ON K1S 5B6, Canada}
\email{wang@math.carleton.ca}

\begin{abstract}
Complete permutation polynomials play an important role in cryptography, particularly in the design of cryptographic primitives such as the Lai--Massey scheme and S-boxes. We generalize a result of Sun, Li, Guo, and Qu (2021) by characterizing the complete permutation behavior of the mapping $\Psi(X)=M(X+\psi(AX))$ over $\F_q^n$,  where $\F_q$ is a finite field of $q$ elements with $q$ being a prime power, $M\in \GL(n, \F_q)$, $\GL(n, \F_q)$ is the general linear group  of order $n$ over $\F_q$, $A_{m \times n}$ is a full-rank matrix over $\F_q$, and $\psi=(\psi_1,\psi_2,\ldots,\psi_n)$ with each component function $\psi_i:\mathbb{F}_q^m\to\mathbb{F}_q$.
Furthermore, we establish criteria for the permutation and complete permutation properties of the mapping $F(X)=T(X+B^tf(AX))$ over $\mathbb{F}_{q}^n$, $f: \F_{q}^{m} \rightarrow \F_{q}^{n-m}$, $T \in \GL(n, \F_q)$, $A_{m \times n}$ and $B_{(n-m)\times n}$ are full-rank matrices over $\F_q$, $B^t$ represents the transpose of the matrix $B$, and $0<m<n$ are integers. These results also generalize an earlier result of Gravel and Panario (2023), who
showed that any arbitrary function $f$ from $\mathbb{F}_q^m$ to $\mathbb{F}_q^{\,n-m}$ can be extended to a bijection over $\mathbb{F}_{q}^n$ through the mapping 
$F(X)=T(X+B^tf(AX))$, under the condition $AB^t=0$. Here we do not impose the restriction that $AB^t=0$.

\end{abstract}
\maketitle

\section{Introduction}\label{S1}
Let $\F_q$ be a finite field with $q$ elements, where $q$ is a prime power, and $\F_q[x]$ represent the ring of polynomials in variable $x$ over $\F_q.$  A polynomial $f\in \F_q[x]$ is a permutation polynomial if the induced map $c \mapsto f(c)$ is a bijection from $\F_q$ to itself. It may be noted that permutation polynomials were first studied by Hermite~\cite{CC} over the finite fields of prime order. Later, Dickson~\cite{LE} investigated permutation polynomials over arbitrary finite fields. A polynomial is a complete permutation polynomial of $\F_q$ if both $f(x)$ and $f(x)+x$ are permutations over $\F_q.$ The notion of complete permutations was introduced by Mann \cite{M} in order to construct orthogonal Latin squares over arbitrary group $G.$  Niedrreiter and Robinson \cite{NR} studied this work over finite fields. After this, a lot of work has been done to study the complete permutation polynomials;  see \cite{BGQZ, BGZ, BW, BW2, BZ, FLWW, L, LLLZ, LWXZ, MP, SV,  TZMZ,  TW, WL, VS, WLHZ, WLHZ2, XC, XLZH} and references therein.

Complete permutation polynomials play a crucial role in several cryptographic schemes. For instance, the absence of the complete permutation in the Lai–Massey scheme renders it highly susceptible to differential cryptanalysis as a probability 1 differential exists for any number of rounds \cite{V}. Therefore, complete permutation polynomials are useful to avoid the propagation of non zero fixed points. Moreover, these polynomials are used in the design of the mode CLOC  \cite{IMGM} to prevent collisions arising from the addition of input and output. In many other mode of operations, one round uses $k$ as a subkey, while another round uses $f(k)$ as the subkey, and the completeness of $f$ ensures the uniform distribution of $f(k)+k$ \cite{CLLSL, GW}.

 In 2023, Gravel and Panario \cite{GP}, introduced a general construction that extends an arbitrary function $f: \F_{q}^{m} \rightarrow \F_{q}^{n-m}$ to a permutation function over $\F_q^n$. Specifically, they showed that mappings of the form
 \begin{equation}\label{e1}
      F(X)=T(X+B^tf(AX))
 \end{equation}
define permutations of $\F_q^n$ under suitable conditions, where $f$ is a function from $\F_q^m$ to $\F_q^{n-m}$, $T \in \GL(n, \F_q)$,  $\GL(n, \F_q)$ is the general linear group of order $n$ over $\F_q$, $A_{m \times n}$ and $B_{(n-m)\times n}$ are full rank matrices over $\F_q$, and $0<m<n$, are integers; see also \cite{GP2024,GPT2027} for some cryptographical properties of this construction. In this construction, they considered that $AB^t=0.$ Then, they studied the equations based on the finite differences and linear forms of $F(X).$ This naturally leads to the question of whether an arbitrary function $f$ admits an extension to a bijective function $F(X)$, given in Equation \eqref{e1}, once the condition $AB^{t}=0$ is no longer imposed.
Here, we extend this framework by relaxing the condition $AB^t=0$. This significantly complicates the structure of the resulting function, requiring new techniques to analyze permutation behavior. Using the AGW criterion \cite{AGW}, we establish conditions under which the mapping given in Equation \eqref{e1} remains a permutation of $\F_q^n$. In this paper, we investigate not only the permutation behavior of $F(X)$ in Equation \eqref{e1}, but also its complete permutation behavior. More precisely, we establish the conditions for both $F(X)$ and $F(X)+X$ to be permutations of $\F_q^n$. 

On the other hand, Sun, Li, Guo and Qu \cite{SLGQ} recently studied the complete permutations over $\F_q^n$ of the form 
\begin{equation}\label{e2}
    \Psi(X)=M(X+\psi(a^tX)),
\end{equation}
where $M \in \GL(n, \F_q)$, $X, a \in \F_q^n$, $\psi=(\psi_1, \psi_2, \ldots, \psi_n)$ such that the $\psi_i$'s are polynomials over $\F_q.$ Specifically, the authors provide a necessary and sufficient condition for $\Psi(X)$ to be a complete permutation polynomial over $\F_q^n$.  In this paper, we generalize the above result and investigate their applications in cryptography. In particular, we characterize the complete permutation behaviour of the function
\begin{equation}\label{e3}
    \Psi(X)=M(X+\psi(AX)),
\end{equation}
where $A_{m \times n}$ is a full rank matrix, $\psi=(\psi_1, \psi_2, \ldots, \psi_n)$ such that the $\psi_i: \F_q^m  \rightarrow \F_q$ is an arbitrary function. We explore the cryptographic relevance of these permutation structures by studying the application of complete permutation polynomials within the Lai--Massey scheme \cite{LM}, a symmetric-key cryptographic construction underlying block ciphers such as IDEA and FOX. This connection provides a framework for applying our permutation results to the analysis and design of cryptographic transformations.

The structure of this paper is organized as follows. Section \ref{S2} introduces some notations and preliminary results, which are essential for the main results of this paper. Our extensions relaxing the orthogonality condition for the linearly extended discrete functions and the study of their complete permutation condition are given in Section \ref{S3}. In Section \ref{S4}, we extend and characterize complete permutations in the setting of \cite{SLGQ}, and explore cryptographical applications of these complete permutation polynomials. The paper is concluded in Section \ref{s5}.

\section{Preliminaries}\label{S2}
Throughout this paper, we use the following notations:
\begin{itemize}
   \item we write $X =(x_1, x_2, \ldots, x_n) \in \F_q^n$ to denote a column vector;
   \item $A=[a_{ij}]$ is an $m \times n$ full rank matrix, where $a_{ij} \in \F_q$;
   \item $B=[b_{ij}]$ is an $(n-m)\times n$ full rank matrix, where $b_{ij} \in \F_q$;
   \item $\mathbf{I_r}$ is the identity matrix of order $r$;
   \item $\mathbf{0_{r \times s}}$ is the zero matrix of order $r \times s$;
   \item $A^t$ represents the transpose of the matrix $A$.
\end{itemize}
\begin{defn}\cite{LNH_1997}
    A system of polynomials $f_1, f_2, \ldots f_r \in \F_{q}[x_1, x_2, \ldots, x_s]$ is \emph{orthogonal} if the system of equations
\begin{equation}
\left\{
\begin{aligned}
f_1(x_1,x_2,\ldots,x_s) &= a_1,\\
f_2(x_1,x_2,\ldots,x_s) &= a_2,\\
 &\vdots \\
f_r(x_1,x_2,\ldots,x_s) &= a_r
\end{aligned}
\right.
\end{equation}
has  $q^{s-r}$ solutions in $\F_q^{s}$ for each tuple $(a_1,a_2,\ldots,a_r) \in \F_q^{r}$.
\end{defn}

\begin{lem}\label{AGW}\cite[the AGW Criterion]{AGW}
    Let $A$, $S$, and $\bar{S}$ be finite sets with $|S| = |\bar{S}|$, and $f:A\to A,
h:S\to \bar{S},
\lambda:A\to S,
\bar{\lambda}:A\to \bar{S}$
be maps satisfying
\(
\bar{\lambda}\circ f = h\circ \lambda
\)
$$\begin{CD}
A @>f>> A\\
  @V{\lambda} VV   @VV \overline{\lambda}. V\\
 S  @>>h>  \overline{S}\
\end{CD}$$
If both $\lambda$ and $\bar{\lambda}$ are surjective, then the following statements are equivalent:

\begin{enumerate}
    \item[(i)] $f$ is bijective (that is, a permutation of $A$);
    
    \item[(ii)] $h$ is bijective from $S$ to $\bar{S}$, and $f$ is injective on $\lambda^{-1}(s)$ for each
    $s\in S$.
\end{enumerate}

\end{lem}
In our work, we need the following results from \cite{GP} and \cite{SLGQ}, respectively.

\begin{lem}\label{l21}\cite[Lemma 1]{GP}
    Let $\F_q$ be a finite field. Then $F(X)=T(X+B^tf(AX))$ is bijective over $\F_q^n.$
\end{lem}
\begin{lem}\label{l22}\cite[Theorem 1]{SLGQ}
Let $a=(a_1,a_2,\ldots,a_n)\in \F_q^n$, $\psi=(\psi_1, \psi_2, \ldots, \psi_n)$ with
polynomials $\psi_1, \psi_2, \ldots, \psi_n$ over $\F_q$, $M\in \F_q^{n\times n}$  a
non-singular matrix, and $\mathbf{I}_n$  the identity matrix of order $n$.
Moreover, let us define $\Psi(X): \F_q^n\to\F_q^n$ as 
\begin{equation*}
\Psi(X)
=
M\Bigl(X+\psi(a^{\,t}X)\Bigr),
\end{equation*}
and $G_1:\F_q\to\F_q$ as
\[
G_1(x)
=x+a^{\,t}f(x)
=x+a_1\psi_1(x)+\cdots+a_n\psi_n(x).
\]
Then $\Psi$ is a complete permutation over
$\F_q^n$ if and only if $G_1$ permutes $\F_q$ and one of the following holds:

\begin{enumerate}

\item
$\operatorname{rank}({M-\mathbf{I}_n})=n-1$,
$a\notin \mathcal{C}(M^t-\mathbf{I}_n)$ and the map
$G_2:\F_q\to\F_q$ defined as
\[
 G_2(x)
=
c^{\,t}M\psi(x)
=
c^{\,t}\psi(x)
\]
permutes $\F_q$, where
$\mathbf{0}_{n\times 1}\neq c\in\F_q^n$ satisfies
\(
(M^t-\mathbf{I}_n)c=\mathbf{0}_{n \times 1}
\)
and $\mathcal{C}(M^t-\mathbf{I}_n)$ is the column space of $M^t-\mathbf{I}_n$;
\item
$\operatorname{rank}(M-\mathbf{I}_n)=n$, and the map
$G_3:\F_q\to\F_q$ defined as
\[
G_3(x)
=
x+c^{\,t}M\psi(x)
=
x+(a+c)^{t}\psi(x)
\]
permutes $\F_q$, where
$c=(M^t-\mathbf{I}_n)^{-1}a$.
\end{enumerate}
\end{lem}
\section{Extended discrete permutations and complete permutations}\label{S3}

In this section, we extend the result established in \cite{GP} by relaxing the condition $AB^t=0$. Furthermore, we establish criteria under which the mapping
\[
F(X)=T\bigl(X+B^t f(AX)\bigr)
\]
over $\mathbb{F}_q^n$ is a complete permutation.
 
\begin{thm}\label{T31}Let $F(X)=T(X+B^tf(AX))$, where $f:\F_q^m \rightarrow \F_q^{n-m}$, $T \in \GL(n, \F_q)$, $A_{m \times n}$, $B_{(n-m)\times n}$ are full rank matrices over $\F_q$, and $0<m<n$, are integers. Then $F(X)$ is bijective if and only if $g(X)=X+AB^t(f(X))$ permutes $\F_{q}^m.$
\end{thm}
\begin{proof}
Since $A$ has full rank $m$, the linear map $A : \F_q^n \to \F_q^m$ is surjective. Consider the following diagram:
$$\begin{CD}
\F_q^n @>F(X)=T(X+B^tf(AX))>> \F_q^n\\
  @V{AX} VV   @VV AT^{-1}X. V\\
 \F_q^m  @>>g(X)=X+AB^tf(X)>  \F_q^m
\end{CD}$$
It is straightforward to verify that the diagram commutes, i.e.,
\[
AT^{-1} \circ F = g \circ A.
\]
Therefore, by the AGW criterion, the mapping $F$ is a permutation of $\F_q^n$ if and only if $F$ is injective on each fiber of $AX$ and $g(X)=X+AB^tf(X)$ is a permutation of $\F_q^m$. It is clear that $F$ is injective on each fiber of $AX$. This completes the proof.
\end{proof}

\begin{thm}\label{T32} Let $F(X)=T(X+B^tf(AX))$ and $f(X+a)-f(X) \in \ker(AB^t)$ for every $0\neq a \in \F_q^m$, where $\ker(AB^t)=\{ Y  \in \F_q^{n-m} \mid AB^tY=0 \}$. Then $F$ is bijective.
\end{thm}
\begin{proof} 
By Theorem~\ref{T31}, the mapping $F$ permutes $\F_q^n$ if and only if 
\[
g(X)=X+AB^tf(X)
\]
permutes $\F_q^m$. We want to show that $g$ is bijective. In contrast, suppose $Y_1, Y_2 \in \F_q^m$ such that 
\[
g(Y_1)=g(Y_2).
\]
Then
\[
Y_1 - Y_2 + AB^t\big(f(Y_1)-f(Y_2)\big)=0.
\]
Let $a = Y_1 - Y_2$. If $a \neq 0$, then by assumption we have 
\[
f(Y_1)-f(Y_2) = f(Y_2+a)-f(Y_2) \in \ker(AB^t),
\]
that is, $AB^t\big(f(Y_1)-f(Y_2)\big)=0$. Hence, the above equation reduces to $a=0$, a contradiction. Therefore, $Y_1=Y_2$, and $g$ is injective. Since $\F_q^m$ is finite, injectivity of $g$ implies bijectivity. Consequently, $g$ permutes $\F_q^m$, and hence $F$ permutes $\F_q^n$.
\end{proof}
In the following theorem, we consider 
\begin{equation}\label{restrictionAB}
AB^t=
\begin{bmatrix}
\mathbf{I}_k & \mathbf{0}_{k \times (n-m-k)} \\
\mathbf{0}_{(m-k) \times k} & \mathbf{0}_{(m-k) \times (n-m-k)}
\end{bmatrix},
\end{equation}
where $k <m,$ $k< n-m$, and extend some non necessarily bijective function $f:\F_q^m \rightarrow \F_q^{n-m}$ to a bijective function $F(X)=T(X+B^tf(AX))$ over $\F_q^n.$
\begin{thm}\label{T33}
     Let $F(X)=T(X+B^tf(AX))$ and $f:\F_q^m \rightarrow \F_q^{n-m}$ be defined as 
     \begin{equation*}
     \begin{split}
     f(x_1, x_2, \ldots, x_m)=&(g_1(x_1, x_2, \ldots, x_k)-x_1, g_2(x_1, x_2, \ldots, x_k)-x_2, \ldots, g_k(x_1, x_2, \ldots, x_k)-x_k, \\&h_1(x_1,x_2 \ldots, x_m), \ldots,  h_{n-m-k}(x_1,x_2 \ldots, x_m)),
     \end{split}
     \end{equation*}
     where $g_i: \F_q^k \rightarrow \F_q$ such that $g_1, g_2, \ldots, g_k$ form an orthogonal system in $\F_q$   and $h_j: \F_q^m \rightarrow \F_q$ is an arbitrary function for $1 \leq i \leq k$, $1\leq j \leq n-m-k.$ Then $F$ permutes $\F_q^n.$
\end{thm}
\begin{proof}
    To prove the result, we invoke Theorem~\ref{T31}. Recall that 
\[
A : \F_q^n \to \F_q^m \quad \text{and} \quad B^t : \F_q^{n-m} \to \F_q^n
\]
satisfy \eqref{restrictionAB}. Then the associated mapping $g : \F_q^m \to \F_q^m$ is given by
\[
g(X)=X+AB^t f(X), \quad X=(x_1,x_2,\ldots,x_m) \in \F_q^m.
\]
Using the form of $AB^t$ given in \eqref{restrictionAB}, we obtain
\[
\begin{aligned}
\begin{split}
g(x_1, x_2, \ldots, x_m)=&(x_1, x_2, \ldots, x_m)+AB^tf(x_1, x_2, \ldots, x_m)
\\=& (x_1, x_2, \ldots, x_m) + AB^t(g_1(x_1, x_2, \ldots, x_k)-x_1, g_2(x_1, x_2, \ldots, x_k)-x_2, \ldots, \\&g_k(x_1, x_2, \ldots, x_k)-x_k, h_1(x_1, x_2, \ldots, x_m), \ldots, 
h_{n-m-k}(x_1, x_2, \ldots, x_m))
\\=&
    (g_1(x_1, x_2, \ldots, x_k), g_2(x_1, x_2, \ldots, x_k), \ldots, g_k(x_1, x_2, \ldots, x_k), x_{k+1}, \ldots, x_m).
\end{split}
\end{aligned}
\]

It follows immediately that $g$ is injective, since $g_1, g_2, \ldots, g_k$ form an orthogonal system in  $\F_q$. As $\F_q^m$ is finite, $g$ is therefore bijective, i.e., a permutation of $\F_q^m$. Consequently, by Theorem~\ref{T31}, the mapping $F$ permutes $\F_q^n$.

\end{proof}
\begin{rmk}
    There are total $q^{(n-m)q^m}$ functions from $\F_q^m$ to $\F_q^{n-m}.$  By Theorem \ref{T33}, we can extend  a total number of $(q^k)!q^{(n-m-k)q^m}$ functions to bijective functions $F$ over $\F_q^n.$ 
\end{rmk}
Next we study the complete permutation behaviour of function $F(X)=T(X+B^tf(AX))$ over $\F_q^n.$

\begin{thm}\label{T34}
    Let $F(X)=T(X+B^tf(AX))$, where $T \in \GL(n, \F_q)$, $A_{m \times n}$, $B_{(n-m)\times n}$ are full rank matrices over $\F_q$, and $0<m<n$, are integers. Then $F$ is a complete permutation over $\F_q^n$ if and only if $g(X)=X+AB^t(f(X))$ permutes $\F_{q}^m$ and one of the following holds:
    \begin{enumerate}
        \item $\operatorname{rank}(T+\mathbf{I_n})=n$ and $$h_1(X)=X+A(T+\mathbf{I_n})^{-1}TB^tf(X)$$ permutes $\F_q^m;$
        \item $\operatorname{rank}(T+\mathbf{I_n})=n-m$, $$h_2(X)=CTB^tf(X)$$ permutes $\F_q^m$ and $r_j \notin span\{R_1, R_2,\ldots, R_n, r_1, r_2, \ldots, r_{j-1}, r_{j+1}, \ldots, r_m\}$, where $R_i$'s are rows of the matrix $T + \mathbf{I_n}$, $r_j$'s are rows of the matrix $A$, $1 \leq i \leq n$, $1\leq j \leq m$, and $C_{m \times n}$ is a full rank matrix  over $\F_q$ such that $(T^t+\mathbf{I_n})C^t=0$;
        \item $n-m<\operatorname{rank}(T+\mathbf{I_n})=k<n$, $CTB^tf(AX)$ is an onto map and $F(X)+X$ is injective on each fiber of $CTB^tf(AX)$, where $C_{(n-k) \times n}$ is a full rank matrix over $\F_q$ such that $(T^t+\mathbf{I_n})C^t=0.$ 
    \end{enumerate}
\end{thm}
\begin{proof}
    From Theorem \ref{T31}, it is clear that $F$ permutes $\F_q^n$ if and only if $g(X)=X+AB^tf(X)$ permutes $\F_q^m$. To determine the completeness of $F$, we study the following two cases.

    \textbf{Case 1.} If $\operatorname{rank}(T+\mathbf{I_n})=n$, then consider the following diagram
$$\begin{CD}
\F_q^n @>(T+\mathbf{I_n})X+TB^tf(AX)>> \F_q^n\\
 @V {AX} VV   @VV A(T+\mathbf{I_n})^{-1}X.  V\\
 \F_q^m  @>>h_1(X)=X+A(T+\mathbf{I_n})^{-1}TB^tf(X)>  \F_q^m
\end{CD}$$
Since, $A$ is a full rank matrix, we have that $A: \F_q^n \rightarrow \F_q^m$ and $A(T+\mathbf{I_n})^{-1}: \F_q^n \rightarrow \F_q^m$ are onto maps. Moreover,
$$h_1\circ A=A(T+\mathbf{I_n})^{-1}\circ (F+\mathbf{I_n}).$$
Hence, by the AGW criterion, $F(X)+X$ permutes $\F_q^n$ if and only if $F(X)+X$ is injective on each fiber of $AX$ and $h_1(X)=X+A(T+\mathbf{I_n})^{-1}TB^tf(X)$ permutes $\F_q^m.$ Since $\operatorname{rank}(T+\mathbf{I_n})=n$, $F(X)+X$ is injective on each fiber of $AX$. Therefore, $F(X)+X$ permutes $\F_q^n$ if and only if $h_1(X)=X+A(T+\mathbf{I_n})^{-1}TB^tf(X)$ permutes $\F_q^m.$

\textbf{Case 2.} We assume that $\operatorname{rank}(T+\mathbf{I_n}) <n.$ If $\operatorname{rank}(T+\mathbf{I_n})\leq n-m$, then $\dim \ker(T+\mathbf{I_n})\geq m.$ Therefore, there exists a full rank matrix $C_{m\times n}$ such that $(T^t+\mathbf{I_n})C^t=0.$ Now, consider the following commutative diagram
$$\begin{CD}
\F_q^n @>(T+\mathbf{I_n})X+TB^tf(AX)>> \F_q^n\\
 @V {AX} VV   @VV CX.  V\\
 \F_q^m  @>>h_2(X)=CTB^tf(X)>  \F_q^m
\end{CD}$$
As $A$ and $C$ are full rank matrices, $AX$ and $CX$ are onto maps from $\F_q^n$ to $\F_q^m.$ By the AGW criterion, $F(X)+X$ permutes $\F_q^n$ if and only if $h_2(X)=CTB^tf(X)$ permutes $\F_q^m$ and $F(X)+X$ is injective on each fiber of $AX.$ Let $S_b=\{X\in \F_q^n \mid AX=b\}.$ We want to check the injectivity of $F(X)+X$ on the set $S_b.$ For this it suffices to show that 
 the following system has a unique solution in $\F_q^n$ for  $e \in \F_q^n$,
\[
\begin{aligned}
 (T+\mathbf{I_n})X+TB^tf(AX)&=e  \\
 AX&=b.
\end{aligned}
\]
From the above system, we have \[
\begin{bmatrix}
T+\mathbf{I_n}  \\
A
\end{bmatrix}X=\begin{bmatrix}
e-TB^tf(b)  \\
b
\end{bmatrix}.\] The system has a unique solution in $\F_q^n$ if and only if  \[
\operatorname{rank} \left(\begin{bmatrix}
T+\mathbf{I_n} \\
A
\end{bmatrix}\right)=n.\] If $\operatorname{rank} (T+\mathbf{I_n})<n-m$, then  \[\operatorname{rank} \left(\begin{bmatrix}
T+\mathbf{I_n}  \\
A
\end{bmatrix}\right ) \neq n.\]   Therefore,  \[\operatorname{rank} \left(\begin{bmatrix}
T+\mathbf{I_n}  \\
A
\end{bmatrix} \right )=n\]
if and only if $\operatorname{rank}(T+\mathbf{I_n})=n-m$ and $r_j \notin span\{R_1, R_2,\ldots, R_n, r_1, r_2, \ldots, r_{j-1}, r_{j+1}, \ldots, r_m\}$, where $R_i$'s are rows of the matrix $T +\mathbf{I_n}$ and $r_j$'s are rows of the matrix $A$, $1\leq i \leq n$, $1 \leq j \leq m.$

Next, if $n-m<\operatorname{rank}(T+\mathbf{I_n})<n.$ Let $\operatorname{rank}(T+\mathbf{I_n})=k$ and $\dim \ker(T+\mathbf{I_n})=n-k$. We consider the following commutative diagram 
$$\begin{CD}
\F_q^n @>(T+\mathbf{I_n})X+TB^tf(AX)>> \F_q^n\\
 @V {CTB^tf(AX)} VV   @VV CX,  V\\
 \F_q^{n-k}  @>>h_3(X)=X>  \F_q^{n-k}
\end{CD}$$
where $C_{(n-k) \times n}$ is a full rank matrix such that $(T^t+\mathbf{I_n})C^t=0.$ If $CTB^tf(AX)$ is not an onto map from $\F_q^n$ to $\F_q^{n-k}$, then $F(X)+X$ will not be permutation over $\F_q^n$. If $F(X)+X$ is permutation, then $C((T+\mathbf{I_n})X+TB^tf(AX))=CTB^tf(AX)$ is onto, which is a contradiction.
 Thus, by the AGW criterion, $F(X)+X$ permutes $\F_q^n$ if and only if $CTB^tf(AX)$ is onto map and $F(X)+X$ is injective on each fiber of $CTB^tf(AX).$
\end{proof}
\begin{exmp} Let $q=2^u$, where $u$ is a positive integer. Let $F(X)=T(X+B^tf(AX))$,
\begin{equation*}
\begin{aligned}
T &=
\begin{bmatrix}
0 & 1 & 0 & 0 & \cdots & 0 & 0 \\
0 & 0 & 1 & 0 & \cdots & 0 & 0 \\
0 & 0 & 0 & 1 & \cdots & 0 & 0 \\
\vdots & \vdots & \vdots & \vdots & \ddots & \vdots & \vdots \\
0 & 0 & 0 & 0 & \cdots & 1 & 0 \\
0 & 0 & 0 & 0 & \cdots & 0 & 1 \\
1 & 0 & 0 & 0 & \cdots & 0 & 1
\end{bmatrix},
\quad \quad A &=
\begin{bmatrix}
1 & 0 & \cdots & 0 \\
0 & 1 & \cdots & 0
\end{bmatrix}, \quad \quad
B &=
\begin{bmatrix}
\mathbf{I}_{n-2} & \mathbf{0}_{n-2 \times 2} \\
\end{bmatrix}
\end{aligned}
\end{equation*}
and $f: \F_q^2 \rightarrow \F_q^{n-2}$ defined as $$f(x_1, x_2)=(\phi(x_2), 0, 0,\ldots, 0),$$ where $\phi(x_2)+x_2$ is a permutation over $\F_q.$

    We have $\operatorname{rank}(T)=n$, $\operatorname{rank}(T+\mathbf{I_n})=n$, $$g(x_1,x_2)=(\phi(x_2)+x_1, x_2)$$ and $$h_1(x_1,x_2)=(x_1+\phi(x_2), x_2+\phi(x_2)).$$ It is easy to check that $g(x_1,x_2)$ and $h_1(x_1,x_2)$ permute $\F_q^2.$ Therefore, by Theorem \ref{T34}, $F$ is a complete permutation over $\F_q^n.$
\end{exmp}

\begin{exmp}
   Let $F(X)=T(X+B^tf(AX))$, 
\begin{equation*}
\begin{aligned}
T &=
\begin{bmatrix}
-1 & 0 & 0 & 0 & \cdots & 0 & 0 \\
0 & 0 & -1 & 0 & \cdots & 0 & 0 \\
0 & 0 & 0 & -1 & \cdots & 0 & 0 \\
\vdots & \vdots & \vdots & \vdots & \ddots & \vdots & \vdots \\
0 & 0 & 0 & 0 & \cdots & -1 & 0 \\
0 & 0 & 0 & 0 & \cdots & 0 & -1 \\
0 & -1 & 0 & 0 & \cdots & 0 & 0
\end{bmatrix},
\quad \quad
A &=
\begin{bmatrix}
1 & 0 & \cdots & 0 \\
0 & 0 & \cdots & 1
\end{bmatrix},
\quad \quad
B &=
\begin{bmatrix}
\mathbf{I}_{n-2} & \mathbf{0}_{(n-2)\times 2}
\end{bmatrix}
\end{aligned}
\end{equation*}
and $f: \F_q^2 \rightarrow \F_q^{n-2}$ defined as $$f(x_1, x_2)=(\phi_1(x_2), \phi_2(x_1), 0,\ldots, 0),$$ where $\phi_1$ and $\phi_2$ are permutations over $\F_q.$

In this case, $\operatorname{rank}(T)=n$, $\operatorname{rank}(T+\mathbf{I_n})=n-2$, \[C=
\begin{bmatrix}
1 & 0 & 0 & \cdots & 0 \\
0 & 1 & 1 & \cdots & 1
\end{bmatrix},
\] $g(x_1,x_2)=(x_1+\phi_1(x_2),x_2)$ and $h_2(x_1,x_2)=(-\phi_1(x_2), -\phi_2(x_1))$. One can easily check that $g(x_1, x_2)$ and $h_2(x_1, x_2)$ are permutations over $\F_q^2.$ Thus, by Theorem \ref{T34}, $F$ is a complete permutation over $\F_q^n.$
\end{exmp}

\begin{exmp}
   Let $F(X)=T(X+B^tf(AX))$, 
\begin{equation*}
\begin{aligned}
T &=
\begin{bmatrix}
-1 & 0 & 0 & 0 & \cdots &0 & 1 & 1 \\
0 & 0 & 1 & 0 & \cdots  & 0& 0 & 0 \\
0 & 0 & 0 & 1 & \cdots & 0 & 0 & 0 \\
\vdots & \vdots & \vdots & \vdots & \ddots & \vdots & \vdots & \vdots\\
0 & 0 & 0 & 0 & \cdots &1  & 0 & 0 \\
0 & 0 & 0 & 0 & \cdots & 0 & 1 & 0 \\
0 & 0 & 0 & 0 & \cdots & 0 & 0 & 1 \\
1 & 1 & 0 & 0 & \cdots & 0 & 0 & 0
\end{bmatrix},
\quad
A &=
\begin{bmatrix}
0 & 0 & \cdots & 1 & 0 \\
0 & 0 & \cdots & 0 & 1
\end{bmatrix},
\quad
B &=
\begin{bmatrix}
1 & 0 & 0 & 0 & 0 &\cdots & 0 \\
0 & 0 & 0 & 1 & 0 &\cdots & 0 \\
0 & 0 & 0 & 0 & 1 &\cdots & 0 \\
0 & 0 & 0 & 0 & 0 &\cdots & 0 \\
\vdots & \vdots & \vdots &\vdots & \vdots & \ddots & \vdots \\
0 & 0 & 0 & 0 & 0 & \cdots & 1 
\end{bmatrix}
\end{aligned}
\end{equation*}
and $f: \F_q^2 \rightarrow \F_q^{n-2}$ defined as $$f(x_1, x_2)=(\phi_1(x_2), \phi_2(x_2), \ldots, \phi_{n-4}(x_2), 0, \alpha x_2),$$ where $\phi_1(x_2)$ is a permutation over $\F_q$ and $\alpha+1 \neq 0.$

We have $\operatorname{rank}(T)=n$, $\operatorname{rank}(T+\mathbf{I_n})=n-1$, \[C=
\begin{bmatrix}
1 & 0 & 0 &\cdots & 0 & -1 & 0 \\
\end{bmatrix},
\]

 $g(x_1, x_2)=(x_1,x_2+\alpha x_2)$ and $CTB^tf(AX)=-\phi_1(x_n)$. It is clear that

\[
\begin{aligned}
F(X)+X=&(x_{n-1}+x_n-\phi_1(x_n)+\alpha x_n, x_2+x_3, x_3+x_4+\phi_2(x_n),  x_4+x_5+\phi_3(x_n), \ldots,\\
&x_{n-2}+x_{n-3}+\phi_{n-4}(x_n), x_{n-1}+x_{n-2}, x_{n-1}+x_n+\alpha(x_n), x_1+x_2+x_n+\phi_1(x_n))
\end{aligned}
\]
is injective on each fiber of $CTB^tf(AX)$. Hence, by Theorem \ref{T34}, $F$ is a complete permutation over $\F_q^n.$
\end{exmp}

\section{Other complete permutations and cryptographical applications}\label{S4}

  In this section, we first demonstrate that our result generalizes the result of Sun, Li, Guo, and Qu (2021). We then explore a cryptographic application involving the Lai--Massey structure.\\

\subsection{Generalizations of Sun, Li, Guo and Qu (2021) construction}

In Lemma \ref{l22}, the authors studied the complete permutations over $\F_q^n$ of the form 
\begin{equation*}
    \Psi(x)=M(X+\psi(a^tX)),
\end{equation*}
and provide a necessary and sufficient condition for $\Psi(X)$ to be a complete permutation polynomial over $\F_q^n.$
In their work, they consider $M \in \GL(n , \F_q)$, $X, a \in \F_q^n$, $\psi=(\psi_1, \psi_2, \ldots, \psi_n)$ such that $\psi_i$'s are polynomials over $\F_q.$ The following theorem generalizes their result, in particular, we characterize the complete permutation behaviour of the function
\begin{equation*}
    \Psi(X)=M(X+\psi(AX))
\end{equation*}
over $\F_q^{n}$, where $A_{m \times n}$ is a full rank matrix, $\psi=(\psi_1, \psi_2, \ldots, \psi_n)$ such that $\psi_i: \F_q^m  \rightarrow \F_q$ is an arbitrary function. 

\begin{thm}\label{T35}
    Let $\Psi(X)=M(X+\psi(AX))$, where $M \in \GL(n, \F_q)$, $\psi=(\psi_1, \psi_2, \ldots, \psi_n)$ such that $\psi_i: \F_q^m \rightarrow \F_q$, $A_{m \times n}$ is a full rank matrix, and $0<m<n$ are integers. Then $\Psi$ is a complete permutation over $\F_q^n$ if and only if $g(X)=X+A(\psi(X))$ permutes $\F_{q}^m$ and one of the following holds:
    \begin{enumerate}
        \item $\operatorname{rank}(M+\mathbf{I_n})=n$ and $$h_1(X)=X+A(M+\mathbf{I_n})^{-1}M\psi(X)$$ permutes $\F_q^m;$
         \item $\operatorname{rank}(M+\mathbf{I_n})=n-m$, $$h_2(X)=CM\psi(X)$$ permutes $\F_q^m$ and $r_j \notin span\{R_1, R_2,\ldots, R_n, r_1, r_2, \ldots, r_{j-1}, r_{j+1}, \ldots, r_m\}$, where $R_i$'s are rows of the matrix $M+ \mathbf{I_n}$, $r_j$'s are rows of the matrix $A$, $1\leq i \leq n$, $1\leq j \leq m$, and $C_{m \times n}$ is a full rank matrix  over $\F_q$ such that $(M^t+\mathbf{I_n})C^t=0$;
        \item $n-m<\operatorname{rank}(M+\mathbf{I_n})=k<n$, $CM\psi(AX)$ is an onto map and $\Psi(X)+X$ is injective on each fiber of $CM\psi(AX)$, where $C_{(n-k) \times n}$ is a full rank matrix over $\F_q$ such that $(M^t+\mathbf{I_n})C^t=0.$  
    \end{enumerate}
\end{thm}
\begin{proof}
    The proof follows similar arguments as those used in Theorem \ref{T34}.   We therefore skip the details. 
\end{proof}
\begin{rmk}
It is important to note that the above generalization encompasses a substantially richer class of functions over $\F_q^n$ than those considered in Lemma~\ref{l22}. In the setting of Lemma~\ref{l22}, $a\in\F_q^n$ and each $\psi_i:\F_q\rightarrow\F_q$ is a polynomial for $1\leq i\leq n$. In contrast, our generalization allows $A$ to be an $m\times n$ matrix and each $\psi_i$ to be an arbitrary function from $\F_q^m$ to $\F_q$. Thus, permitting $\psi_i:\F_q^m\rightarrow\F_q$ significantly increases the flexibility in the choice of the component functions and, consequently, provides a broader range of possibilities for constructing complete permutation polynomials over $\F_q^n$ than the more restrictive choice with $m=1$. 
\end{rmk}

Linear complete permutation polynomials play an important role in various cryptographic constructions; see, for instance, \cite{AC, LLG}. Motivated by these applications, in the following theorem, we present a class of linear complete permutation polynomials obtained by considering the function $\psi=(\psi_1, \psi_2, \ldots, \psi_n)$ such that $\psi_i: \F_q^m \rightarrow \F_q$ is defined as $\psi_i=d_{i1}x_1+d_{i2}x_2+\cdots+d_{im}x_m$ for all $1\leq i \leq m$ and $\psi_i=0$ for all $i>m.$
\begin{thm}\label{T36}
     Let $\Psi(X)=M(X+\psi(AX))$, where $M \in \GL(n, \F_q)$, $\psi=(\psi_1, \psi_2, \ldots, \psi_n)$ such that $\psi_i: \F_q^m \rightarrow \F_q$, $A_{m \times n}$ is a full rank matrix, and $0<m<n$ are integers.  Then $\Psi$ is a complete permutation over $\F_q^n$ if and only if $\mathbf{I_m}+A_1D$  is invertible over $\F_q$, where $A_1=[a_{ij}]_{m \times m}$, $D=[d_{ij}]_{m \times m}$, $a _{ij}, d_{ij} \in \F_q$, $1 \leq i, j \leq m$, and one of the following holds:
     \begin{enumerate}
        \item $\operatorname{rank}(M+\mathbf{I_n})=n$ and $\mathbf{I_m}+A(M+\mathbf{I_n})^{-1}MD_1$ is invertible over $\F_q$, where 
        \[D_1=\begin{bmatrix}
D \\
\mathbf{0}_{(n-m) \times m}
\end{bmatrix};\]
 \item $\operatorname{rank}(M+\mathbf{I_n})=n-m$, $CMD_1$ is invertible over $\F_q$ and 
 \(r_j \notin span\{(R_1, R_2,\ldots, R_n, r_1, r_2,
 \ldots, \\ r_{j-1}, r_{j+1}, \ldots, r_m\}\), 
 where $R_i$'s are rows of the matrix $M+ \mathbf{I_n}$, $r_j$'s are rows of the matrix $A$, $1\leq i \leq n$, $1 \leq j \leq m$, and $C_{m \times n}$ is a full rank matrix over $\F_q$ such that $(M^t+\mathbf{I_n})C^t=0$;
        \item $n-m<\operatorname{rank}(M+\mathbf{I_n})=k<n$, $CM\psi(AX)$ is an onto map and $\Psi(X)+X$ is injective on each fiber of $CM\psi(AX)$, where $C_{(n-k) \times n}$ is a full rank matrix over $\F_q$ such that $(M^t+\mathbf{I_n})C^t=0.$ 
    \end{enumerate}
\end{thm}
\begin{proof}
    By Theorem \ref{T35}, $\Psi$ is a permutation over $\F_q^n$ if and only if $g(X)=X+A(\psi(X))$ permutes $\F_q^m.$ We have 
    \begin{equation*}
    \begin{split}
    g(x_1,x_2, \ldots, x_m)=&(x_1,x_2, \ldots, x_m)+A(\psi(x_1, x_2, \ldots, x_m))
    \\=&(x_1,x_2, \ldots, x_m)+A(\psi_1(x_1, x_2, \ldots, x_m), \psi_2(x_1, x_2, \ldots, x_m), \ldots, \psi_n(x_1, x_2, \ldots, x_m))\\=&
        (x_1,x_2, \ldots, x_m)+A(d_{11}x_1+d_{12}x_2+\cdots+d_{1m}x_m,
         d_{21}x_1+d_{22}x_2+\cdots+d_{2m}x_m, \\&\ldots, d_{m1}x_1+d_{m2}x_2
        +\cdots+d_{mm}x_m, 0, 0, \ldots, 0)\\
    =&(x_1,x_2, \ldots, x_m)+A\begin{bmatrix}
D \\
\mathbf{0}_{(n-m) \times m}
\end{bmatrix}(x_1,x_2, \ldots, x_m)
\\=&(x_1,x_2, \ldots, x_m)+A_1D(x_1,x_2, \ldots, x_m)
\\=&(\mathbf{I_m}+A_1D)(x_1,x_2, \ldots, x_m).
\end{split}
\end{equation*}
Therefore, $g$ permutes $\F_q^m$ if and only if $\mathbf{I_m}+A_1D$ is invertible over $\F_q.$ Furthermore, if $\operatorname{rank}(M+\mathbf{I_n})=n$, then $\Psi(X)+X$ permutes $\F_q^n$ if and only if $h_1(X)=X+A(M+\mathbf{I_n})^{-1}M\psi(X)$ permutes $\F_q^m.$ We have
\begin{equation*}
 \begin{split}   
h_1(x_1, x_2, \ldots, x_m)=&(x_1, x_2, \ldots, x_m)+A(M+\mathbf{I_n})^{-1}M\psi(x_1, x_2, \ldots, x_m)
\\=&(x_1, x_2, \ldots, x_m)+A(M+\mathbf{I_n})^{-1}M(\psi_1(x_1, x_2, \ldots, x_m), \\& \psi_2(x_1, x_2, \ldots, x_m), \ldots, \psi_n(x_1, x_2, \ldots, x_m))
 \\=&
        (x_1,x_2, \ldots, x_m)+A(M+\mathbf{I_n})^{-1}M(d_{11}x_1+d_{12}x_2+\cdots+d_{1m}x_m,\\& d_{21}x_1+d_{22}x_2+\cdots+d_{2m}x_m, \ldots, d_{m1}x_1+d_{m2}x_2
        +\cdots+d_{mm}x_m, 0, 0, \ldots, 0)
\\=&(x_1, x_2, \ldots, x_m)+A(M+\mathbf{I_n})^{-1}MD_1(x_1, x_2, \ldots, x_m)
\\=&(\mathbf{I_m}+A(M+\mathbf{I_n})^{-1}MD_1)(x_1, x_2, \ldots, x_m).
\end{split}
\end{equation*}
Hence, $h_1(X)=X+A(M+\mathbf{I_n})^{-1}M\psi(X)$ permutes $\F_q^m$ if and only if $\mathbf{I_m}+A(M+\mathbf{I_n})^{-1}MD_1$ is invertible over $\F_q.$

If $\operatorname{rank}(M+\mathbf{I_n})=n-m$ and $r_j \notin span\{R_1, R_2,\ldots, R_n, r_1, r_2, \ldots, r_{j-1}, r_{j+1}, \ldots, r_m\}$, where $R_i$'s are rows of the matrix $M+ \mathbf{I_n}$ and $r_j$'s are rows of the matrix $A$, $1\leq i \leq n$, $1\leq j \leq m$, and $C_{m \times n}$ is a full rank matrix such that $(M^t+\mathbf{I_n})C^t=0$, then $\Psi(X)+X$ is permutation over $\F_q^n$ if and only if $h_2(X)=CM\psi(X)$ permutes $\F_q^m.$ We have
\begin{equation*}
\begin{split}
h_2(x_1, x_2, \ldots, x_m)=&CM\psi(x_1, x_2, \ldots, x_m)
\\=&CM(\psi_1(x_1, x_2, \ldots, x_m), \psi_2(x_1, x_2, \ldots, x_m), \ldots, \psi_n(x_1, x_2, \ldots, x_m))
\\=&CM(d_{11}x_1+d_{12}x_2+\cdots+d_{1m}x_m, d_{21}x_1+d_{22}x_2+\cdots+d_{2m}x_m, \ldots, \\&d_{m1}x_1+d_{m2}x_2+\cdots+d_{mm}x_m, 0, 0, \ldots, 0)\\=&CMD_1(x_1, x_2, \ldots, x_m).
\end{split}
\end{equation*}
Consequently, $h_2(X)=CM\psi(X)$ permutes $\F_q^m$ if and only if $CMD_1$ is invertible over $\F_q.$

If $n-m<\operatorname{rank}(M+\mathbf{I_n})<n$, then result follows from Theorem \ref{T35}.
\end{proof}

\begin{exmp}Let $q=2^u$, where $u$ is a positive integer. Let $\Psi(X)=M(X+\psi(AX))$, where
\begin{equation*}
\begin{aligned}
M &=
\begin{bmatrix}
0 & 1 & 0 & 0 & \cdots & 0 & 0 \\
0 & 0 & 1 & 0 & \cdots & 0 & 0 \\
0 & 0 & 0 & 1 & \cdots & 0 & 0 \\
\vdots & \vdots & \vdots & \vdots & \ddots & \vdots & \vdots \\
0 & 0 & 0 & 0 & \cdots & 1 & 0 \\
0 & 0 & 0 & 0 & \cdots & 0 & 1 \\
1 & 0 & 0 & 0 & \cdots & 0 & 1
\end{bmatrix},
\quad \quad
A &=
\begin{bmatrix}
1 & 0 & \cdots & 0 \\
0 & 1 & \cdots & 0
\end{bmatrix},
\quad \quad
\end{aligned}
\end{equation*}
$\psi=(\psi_1, \psi_2, \ldots, \psi_n)$, $\psi_1=a_1x_1+a_2x_2+a_3x_3$, $\psi_2=b_1x_1+b_2x_2+b_3x_3$, $\psi_3=c_1x_1+c_2x_2+c_3x_3$, and $\psi_i=0$ for all $i>3.$ Then $\Psi$ is a complete permutation over $\F_q^n$ if and only if 
$$\displaystyle \sum_{j=1}^{3}\sum_{k=1}^{3}\left(a_kb_{6-(j+k)}\right)c_j \neq a_2b_1+a_1b_2+a_3c_1+b_3c_2+a_1c_3+b_2c_3+a_1+b_2+c_3+1$$ and 
   \begin{equation*}
   \begin{split}
    \displaystyle \sum_{j=1}^{3}\sum_{k=1}^{3}\left(a_kb_{6-(j+k)}\right)c_j \neq a_3b_1+a_1b_3+a_2c_1+a_3c_1+a_1c_2+a_3c_2+b_3c_2+a_1c_3+a_2c_3+b_2c_3+a_1\\+a_3+b_3+c_2+c_3+1.
    \end{split}
    \end{equation*}
\end{exmp}
\begin{proof}
    Since $\operatorname{rank}(M)=rank(M+\mathbf{I_n})=n$, the result follows from Theorem \ref{T36}.
%
\end{proof}
\begin{exmp}
    Let $\Psi(X)=M(X+\psi(AX))$, where
\begin{equation*}
\begin{aligned}
T &=
\begin{bmatrix}
-1 & 0 & 0 & 0 & \cdots & 0 & 0 \\
0 & 0 & -1 & 0 & \cdots & 0 & 0 \\
0 & 0 & 0 & -1 & \cdots & 0 & 0 \\
\vdots & \vdots & \vdots & \vdots & \ddots & \vdots & \vdots \\
0 & 0 & 0 & 0 & \cdots & -1 & 0 \\
0 & 0 & 0 & 0 & \cdots & 0 & -1 \\
0 & -1 & 0 & 0 & \cdots & 0 & 0
\end{bmatrix},
\quad \quad
A &=
\begin{bmatrix}
1 & 0 & \cdots & 0 \\
0 & 0 & \cdots & 1
\end{bmatrix},
\\[1em]
\end{aligned}
\end{equation*}
$\psi=(\psi_1, \psi_2, \ldots, \psi_n)$, $\psi_1=a_1x_1+a_2x_2$, $\psi_2=b_1x_1+b_2x_2$, and $\psi_i=0$ for all $i>2.$ Then $\Psi$ is a complete permutation over $\F_q^n$ if and only if $a_1+1 \neq 0$ and $a_1b_2-a_2b_1 \neq 0.$
\end{exmp}
\begin{proof}
    Since $\operatorname{rank}(M)=n$ and $\operatorname{rank}(M+\mathbf{I_n})=n-2$, the result follows from Theorem \ref{T36}.
\end{proof}

\subsection{A cryptographical application}

The Lai--Massey scheme is a symmetric-key cryptographic construction used in the design of block ciphers, with IDEA and FOX being prominent examples of Lai--Massey-type structures. Vaudenay \cite{V} established the fundamental role of complete permutations (orthomorphisms) in the Lai--Massey scheme by requiring a permutation $\sigma$ such that $X\mapsto \sigma(X)-X$ is also a permutation. Subsequently, Luo, Lai and Gong \cite{LLG}  employed an explicit linear orthomorphism, $\sigma(X,Y)=(Y,X\oplus Y)$, in their pseudorandomness analysis of the Lai--Massey scheme, and studied the role of orthomorphisms in the extended Lai--Massey construction. Later, Aragona and Civino \cite{AC} investigated Lai--Massey schemes with linear orthomorphisms from an algebraic and group-theoretic perspective. These works highlight that complete permutations are a fundamental component of Lai--Massey constructions, although the use of broader classes of explicit complete permutation polynomials over finite fields remains relatively unexplored. Now we give the explicit Lai-Massey structure by using our complete permutations.
Let $P$ be a plaintext of $2n$ symbols over $\F_q$, or
$2n\log_2 q$ bits when $q$ is a power of two, and divide into two elements of $\F_{q}^n$:
\[
    P=(L_0,R_0)\in \F_{q}^n\times \F_{q}^n.
\]
Let
\[
    F_{K_i}: \F_{q}^n \longrightarrow \F_{q}^n
\]
be the keyed round function in round $i$. The function $F_{K_i}$ need not be
invertible. Given the input $(L_{i-1},R_{i-1})$ to round $R_{K_i}$, calculate
\begin{align*}
    \Delta_i &= L_{i-1}+R_{i-1}, \\
    T_i &= F_{K_i}(\Delta_i), \\
    U_i &= L_{i-1}+T_i, \\
    V_i &= R_{i-1}+T_i,\\
    L_i &= \Psi (U_i)
         =M\bigl(U_i+\psi(AU_i)\bigr), \\
    R_i &= V_i. 
\end{align*}
Thus the $i$th round $R_{K_i}$ is the permutation
\begin{equation*}
\boxed{
\begin{aligned}
\mathcal R_{K_i}(L,R)
  =\Bigl(&M\bigl(L+F_{K_i}(L+R)
       +\psi(A(L+F_{K_i}(L+R)))\bigr),\\
       &R+F_{K_i}(L+R)\Bigr).
\end{aligned}}
\end{equation*}
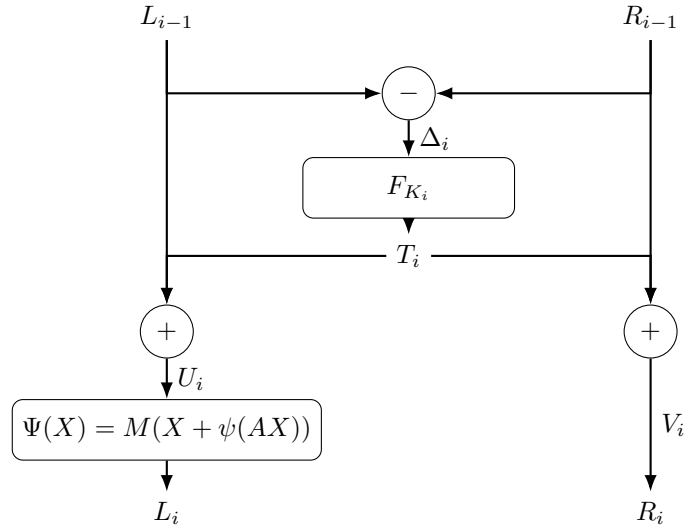
\begin{figure}[H]
\centering
\begin{tikzpicture}[
    >=Latex,
    line/.style={draw,thick,-{Latex[length=2.2mm]}},
    box/.style={draw,rounded corners,minimum width=2.8cm,
                minimum height=8mm,align=center},
    op/.style={draw,circle,minimum size=7mm,inner sep=0pt,fill=white},
    every node/.style={font=\small}
]
\node (L0) at (-3.2,0) {$L_{i-1}$};
\node (R0) at ( 3.2,0) {$R_{i-1}$};
\node[op] (minus) at (0,-1.0) {$-$};
\node[box] (F) at (0,-2.25) {$F_{K_i}$};
\node (T) at (0,-3.15) {$T_i$};
\node[op] (plusL) at (-3.2,-4.15) {$+$};
\node[op] (plusR) at ( 3.2,-4.15) {$+$};
\node[box,minimum width=3.5cm] (sigma) at (-3.2,-5.45)
    {$\Psi(X)=M(X+\psi(AX))$};
\node (L1) at (-3.2,-6.55) {$L_i$};
\node (R1) at ( 3.2,-6.55) {$R_i$};
\draw[line] (L0) -- (plusL);
\draw[line] (R0) -- (plusR);
\draw[line] (L0) |- (minus.west);
\draw[line] (R0) |- (minus.east);
\draw[line] (minus) -- node[right] {$\Delta_i$} (F);
\draw[line] (F) -- (T);
\draw[line] (T) -| (plusL);
\draw[line] (T) -| (plusR);
\draw[line] (plusL) -- node[right] {$U_i$} (sigma);
\draw[line] (sigma) -- (L1);
\draw[line] (plusR) -- node[right] {$V_i$} (R1);
\end{tikzpicture}
\caption{One Lai--Massey encryption round using the complete permutation
from $\Psi$.}
\label{fig:enc-round}
\end{figure}

For an $r$-round cipher, repeat $R_{K_i}$ round for $i=1,2,\ldots,r$, that is,
\[
\mathcal R_{K_r}\circ\mathcal R_{K_{r-1}}
      \circ\cdots\circ\mathcal R_{K_1},
\]
and the ciphertext is
\(
    C=(L_r,R_r).
\)
\begin{figure}[H]
\centering
\begin{tikzpicture}[
    >=Latex,
    line/.style={draw,thick,-{Latex[length=2.2mm]}},
    round/.style={draw,rounded corners,minimum width=2.2cm,
                  minimum height=1cm,align=center},
    every node/.style={font=\small}
]
\node (P) {$P=(L_0,R_0)$};
\node[round,right=8mm of P] (r1) {Round $1$\\$F_{K_1},\Psi$};
\node[round,right=8mm of r1] (r2) {Round $2$\\$F_{K_2},\Psi$};
\node[right=8mm of r2] (dots) {$\cdots$};
\node[round,right=8mm of dots] (rr) {Round $r$\\$F_{K_r},\Psi$};
\node[right=8mm of rr] (C) {$C=(L_r,R_r)$};
\draw[line] (P)--(r1);
\draw[line] (r1)--(r2);
\draw[line] (r2)--(dots);
\draw[line] (dots)--(rr);
\draw[line] (rr)--(C);
\end{tikzpicture}
\caption{The complete $r$-round Lai--Massey encryption scheme.}
\end{figure}
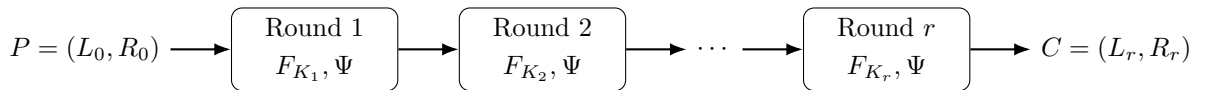
The Lai--Massey structure presented above together with our complete permutation polynomials has potential applications in the construction of cryptographic schemes.

\section{Conclusion}\label{s5}
 We studied the complete permutation polynomials over $\F_q^n$. In particular, we investigated a class of mappings $F(X)=T(X+B^tf(AX))$ without imposing the restriction that $AB^t=0$ as was considered in the work of Gravel and Panario (2023). We established a criteria for the permutation and complete permutation properties of the mapping $F(X)=T(X+B^tf(AX))$ over $\mathbb{F}_{q}^n$. Furthermore, we generalized a result of Sun, Li, Guo, and Qu (2021) by characterizing the complete permutation behavior of the mapping $\Psi(X)=M(X+\psi(AX))$ over $\F_q^n$.  In addition, we demonstrated how the complete permutation polynomials obtained in this work can be incorporated into the Lai--Massey scheme for cryptographic applications.

\section*{Acknowledgments}
  Ramandeep Kaur acknowledges Carleton University for providing all the facilities and research environment during her visit to Carleton University where the major part of the work was carried out.
Ramandeep Kaur and Hridesh Kumar are supported by the Prime Minister’s Research Fellowship, under PMRF IDs 3003658 and 3002900, respectively, at IIT Jammu.
Daniel Panario and Qiang Wang were supported by NSERC (Canada) under Grant No.~RGPIN-2024-05342 and Grant No.~RGPIN-2023-04673, respectively.


\begin{thebibliography}{99}
\bibitem{AGW} A. Akbary, D. Ghioca and Q. Wang, On constructing permutations of finite fields, Finite Fields Appl. 17 (2011), 51--67.
\bibitem{AC} R. Aragona and R. Civino, On invariant subspaces in the Lai–Massey scheme and a primitivity reduction, Mediterr. J. Math. 18 (4) (2021), 165.
\bibitem{BGQZ}  D. Bartoli, M. Giulietti, L. Quoos and G. Zini, Complete permutation polynomials from exceptional polynomials, J. Number Theory 176 (2017), 46--66.
\bibitem{BGZ}  D. Bartoli, M. Giulietti and G. Zini, On monomial complete permutation polynomials, Finite Fields Appl. 41 (2016), 132--158.
\bibitem{BZ} L. A. Bassalygo and V. A. Zinoviev, Permutation and complete permutation polynomials, Finite Fields Appl. 33 (2015), 198--211.
\bibitem{BW} A. Bors and Q. Wang, Cycle types of complete mappings of finite fields, J. Algebra 591 (2022), 577--610.
\bibitem{BW2}	A. Bors and Q. Wang,  Composition and parities of complete mappings of orthomorphism, J. Combin. Theory Ser. A 196 (2023), 105723.
\bibitem{CLLSL}S. Chen, R. Lampe, J. Lee, Y. Seurin and J. Steinberger, Minimizing the two-round Even–Mansour cipher, J. Cryptol. 31 (4) (2018), 1064--1119.
\bibitem{LE} L. E. Dickson, The analytic representation of substitutions on a power of a prime number of letters with a discussion of the linear group,  Ann. Math. 11 (1896), 65--120.
\bibitem{FLWW}   X. Feng,  D. Lin,  L. Wang and  Q. Wang,  Further results on complete permutation monomials over finite fields, Finite Fields Appl. 57 (2019), 47--59.
\bibitem{GP} C. Gravel and D. Panario, Feedback linearly extended discrete functions, J. Algebra  Appl. 22 (2023), 1--15.
\bibitem{GP2024} C. Gravel and D. Panario, Revisiting linearly extended discrete functions, J. Math. Cryptol. 18 (2024), 20240010.
\bibitem{GPT2027} C. Gravel, D. Panario and H. Teixeira, Cycle structure of discrete functions extended through linear transformations, Finite Fields Appl. 117 (2027), 102878.
\bibitem{GW}C. Guo and L. Wang, Revisiting key-alternating Feistel ciphers for shorter keys and multi-user security, in Advances in Cryptology
(Lecture Notes in Comput. Sci.), vol. 11272, T. Peyrin and S. D. Galbraith, Eds. Brisbane, QLD, Australia: Springer, Dec. 2018, pp. 213--243.
\bibitem{CC} C. Hermite, Sur les fonctions de sept lettres, C. R. Acad. Sci. Paris 57 (1863), 750--757.
\bibitem{IMGM}T. Iwata, K. Minematsu, J. Guo and S. Morioka, CLOC: Authenticated encryption for short input, in Fast Software Encryption (Lecture Notes in Comput. Sci.), vol. 8540, C. Cid and
C. Rechberger, Eds. London, U.K.: Springer, Mar. 2014, pp. 149--167.
\bibitem{LM} X. Lai and J. L. Massey, A proposal for a new block encryption standard, in Workshop on the Theory and Application of of Cryptographic Techniques,  Vol 473, Berlin, Heidelberg: Springer, 1990,  pp. 389--404.
\bibitem{L}K. Li, Constructions of complete permutations in multiplication, Des. Codes Cryptogr. 93 (6) (2025), 2205--2228.
\bibitem{LLLZ}L. Li, C. Li, C. Li and X. Zeng, New classes of complete permutation polynomials, Finite Fields Appl. 55 (2019), 177--201.
\bibitem{LWXZ}L. Li, Q. Wang, Y. Xu and X. Zeng,  Several classes of complete permutation polynomials with Niho exponents, Finite Fields Appl.  72 (2021), 101831.
\bibitem{LNH_1997} R. Lidl and H. Niederreiter, Finite Fields, 2nd edn., Encyclopedia Math. Appl., Vol. 20, Cambridge University Press, Cambridge, 1997.
\bibitem{LLG} Y. Luo, X. Lai and Z. Gong, Pseudorandomness Analysis of the (Extended) Lai–Massey Scheme, Inform. Process. Lett. 111 (2) (2010), 90--96.
\bibitem{M}H. B. Mann, The construction of orthogonal Latin squares, Ann. Math. Stat. 13 (4) (1942), 418--423.
\bibitem{MP} A. Muratović-Ribić and E. Pasalic, A note on complete polynomials over finite fields and their applications in cryptography, Finite Fields Appl. 25 (2014), 306--315.
\bibitem{NR} H. Niederreiter and K. H. Robinson, Complete mappings of finite fields, J. Austral. Math. Soc. Ser. A 33 (2) (1982), 197–212. 
\bibitem{SV} R. P. Singh and C. Kumar Vishwakarma, Further results on complete permutation polynomials, Adv. Math. Commun. 26 (2026), 106--124.
\bibitem{SLGQ}  B. Sun, K. Li, J. Guo and L. Qu, New constructions of complete permutations, IEEE Trans. Inform. Theory 67 (11) (2021), 7561--7567.
\bibitem{TZMZ} Z. Tu, X. Zeng, J. Mao and J. Zhou, Several classes of complete permutation polynomials over finite fields of even characteristic, Finite Fields Appl. 68 (2020), 101737.
\bibitem{TW}  A. Tuxanidy and Q. Wang, Compositional inverses and complete mappings over finite fields, Discrete Appl. Math. 217 (2017),  318--329.
\bibitem{V} S. Vaudenay, On the Lai-Massey scheme, in Proc. Int. Conf. Theory Appl. Cryptol. Inf. Secur. Berlin, Germany: Springer, 1999, pp. 8--19.
\bibitem{VS}  C. K. Vishwakarma and R. P. Singh, Some results on complete permutation polynomials and mutually orthogonal Latin squares, Finite Fields Appl. 93 (2024), 102320.
\bibitem{WL} B. Wu and D. Lin, On constructing complete permutation polynomials over finite fields of even characteristic, Discrete Appl. Math. 184 (2015), 213--222.  
\bibitem{WLHZ} G. Wu, N. Li, T. Helleseth and Y. Zhang, Some classes of monomial complete permutation polynomials over finite fields of characteristic two, Finite Fields Appl. 28 (2014), 148--165.  
\bibitem{WLHZ2} G. Wu, N. Li, T. Helleseth and Y. Zhang, Some classes of complete permutation polynomials over $\mathbb{F}_q$, Sci. China Math. 58 (10) (2015),  2081--2094. 
\bibitem{XC}  G. Xu  and X. Cao,  Complete permutation polynomials over finite fields of odd characteristic, Finite Fields Appl. 31 (2015), 228--240.
\bibitem{XLZH}X. Xu, C. Li, X. Zeng and T. Helleseth, Constructions of complete permutation polynomials, Des. Codes Cryptogr. 86 (12) (2018), 2869--2892.
\bibitem{ZWFS}W. Zhang, W. Wu, D. Feng and B. Su, Some new observations on the SMS4 block cipher in the Chinese WAPI standard, in Information Security Practice and Experience (Lecture Notes in Comput. Sci.), vol. 5451, F. Bao, H. Li, and G. Wang, Eds. Xi’an, China: Springer, 2009, pp. 324--335.







 \end{thebibliography}
\end{document}